%% file: kshot.tex
\documentclass[a4paper,11pt]{article}
\include{Settings}

\begin{document}

\title{$k$-shot Broadcasting in Ad Hoc Radio Networks}
\author{Sushanta Karmakar\textsuperscript{1} \and Paraschos Koutris\textsuperscript{2}  \and Aris Pagourtzis\textsuperscript{3}\and Dimitris Sakavalas\textsuperscript{3}}

\maketitle

\begin{center}
 {
 \textsuperscript{1}Indian Institute of Technology, Guwahati (IITG)
 Department of Computer Science and Engineering, Amingaon, North Guwahati, India.\\
 \small{\tt{{sushantak@iitg.ernet.in}}}
 
 \textsuperscript{2}Department of Computer Sciences, University of Wisconsin-Madison. \\
 \small{\tt{{paris@cs.wisc.edu }}}
 
 \textsuperscript{3}School  of Electrical and Computer Engineering \\ National Technical University of Athens, 15780 Athens, Greece.\\
\small{\tt{{pagour@cs.ntua.gr}, {sakaval@corelab.ntua.gr}}}

}

\end{center}





\begin{abstract}
We study distributed broadcasting protocols with few transmissions (`shots') in radio networks where the topology is unknown. In particular, we examine the case in which a bound $k$ is given and a node may transmit at most $k$ times during the broadcasting protocol. Initially, we focus on oblivious algorithms for $k$-shot broadcasting, that is, algorithms where each node decides whether to transmit or not with no consideration of the transmission history. Our main contributions are (a) a lower bound of $\Omega(n^2/k)$ on the broadcasting time of any oblivious $k$-shot broadcasting algorithm and (b) an oblivious broadcasting protocol that achieves a matching upper bound, namely $O(n^2/k)$, for every $k \le \sqrt{n}$ and an upper bound of $O(n^{3/2})$ for every $k > \sqrt{n}$. We also study the general case of adaptive broadcasting protocols where  nodes decide whether to transmit based on all the available information, namely the transmission history known by each. We prove a lower bound of $\Omega\left(n^{\frac{1+k}{k}}\right)$ on the broadcasting time of any protocol by introducing the \emph{transmission tree} construction which generalizes previous approaches. 
%
\end{abstract}
\vspace{.1in}


\section{Introduction}

Energy efficiency has become a central issue in wireless networks, due to the constantly increasing use of autonomous devices with limited power resources. A lot of recent research focuses on how to accomplish communication tasks in an energy-efficient manner without compromising the system performance too much. Much of the work so far has been devoted to the problem of adjusting the transmission ranges of nodes so that the energy cost is minimized.

However, if nodes transmit at a fixed power level it makes sense to consider the number of transmissions as an energy consumption measure. Such a study was initiated by Gasieniec \textsl{et al.}\ in~\cite{GKKPS08}, where broadcasting protocols with few transmissions (`shots') per node were considered for radio networks with known topology. Here, we study the problem in unknown radio networks, that is, networks in which nodes have no knowledge of the topology of the network.

We assume that a bound $k$ is given and a node may transmit at most $k$ times during the broadcasting protocol (\textit{$k$-shot broadcasting}); note that the bound $k$ may well represent the number of transmissions that the power supply of a node can handle. We also assume that the communication is \emph{synchronous}, that is, nodes may  transmit or receive simultaneously; each such simultaneous transmission is called a communication \emph{step}. At each step a node may decide to act either as a \emph{transmitter} or a \emph{receiver}. Whenever a node transmits all its neighbors receive the message. If, however, two neighbors of a node $v$ transmit simultaneously then a \emph{collision} occurs and $v$ receives no message. 

We examine in particular the task of \textit{broadcasting}. In the beginning, there is a unique node (the \emph{source} node) which holds a message $m$, and the goal of a broadcasting protocol is to disseminate $m$ to every node of the graph in a minimum number of steps.

We consider two types of protocols: \emph{adaptive} and \emph{oblivious} protocols; the former refers to protocols where a node may decide whether to transmit or not by taking into account any information it has received during the previous steps, while the latter term refers to protocols where a node makes transmission decisions with very no consideration of the transmission history. Even though adaptive protocols are more powerful, oblivious algorithms are much easier to implement and demand minimal processing time for each node.

By taking into account both models, we study the way in which the limitation on the number of transmissions interacts with the time complexity of broadcasting protocols. 

\paragraph{Related work.}

Distributed broadcasting in radio networks of unknown topology with no limitation in the number of shots has been extensively studied in the literature.

The problem was first introduced by Chlamtac and Kutten~\cite{CK85}. Bar-Yehuda, Goldreich and Itai~\cite{BGI87} gave the first randomized protocol, which completes broadcasting in $O(D\log n +\log^2 n)$ expected time when applied to graphs with $n$ nodes and diameter $D$. Several papers followed~\cite{CR03,KP03} that led to a tight upper bound of $O(D\log (n/D) + \log^2 n)$.

As for the deterministic case, a lower bound of $\Omega(n\log n)$ for general networks was given by Brusci and Del Pinto in \cite{BP97}, improved (for small $D$) to $\Omega(n\log D)$ by~Clementi \textsl{et al.}~\cite{CMS03}. Chlebus {\it et al.}~\cite{CGGPR00} gave the first broadcasting protocol of sub-quadratic time complexity $O(n^{11/6})$. This bound was later improved to $O(n^{5/3}\log^3n)$ by De~Marco and Pelc \cite{MP01} and then by Chlebus {\it et al.}~\cite{CGOR00}, who gave an algorithm with time complexity $O(n^{3/2})$ based on finite geometries. Chrobak, G\c{a}sieniec and Rytter~\cite{CGR00} further improved the bound to $O(n\log^2n)$. Finally, De~Marco~\cite{M08} gave the best currently known upper bound of $O(n\log n \log\log n)$, thus leaving a sub-logarithmic gap between the upper and lower bound.

It should be mentioned that all the aforementioned algorithms with time complexity better than $O(n^{3/2})$ are non-constructive. The best constructive bound so far is that of Indyk~\cite{I02} who presented a somewhat slower constructive version of the protocol of~\cite{CGR00}, achieving $O(n\log^{O(1)}n)$ time complexity. It is also noteworthy that all algorithms proposed so far for deterministic distributed broadcasting in directed graphs are oblivious; this seems to be inherently related to the fact that the nodes have no knowledge about the graph topology.

For undirected networks, Chlebus {\it et al.}~\cite{CGGPR00} gave a deterministic $O(n)$-time broadcasting algorithm with spontaneous wake-up. In case the nodes do not use spontaneous wake-up, an optimal $O(n \log n)$-time broadcasting algorithm was presented in~\cite{KP03}.    

As mentioned above, broadcasting with a limited number of shots was first proposed in~\cite{GKKPS08}. It has also been considered in~\cite{KP09}, where randomized algorithms were proposed; in both cases, only broadcasting in known networks was studied. Another approach to limiting the number of shots was presented in~\cite{BCH09}, where the authors construct algorithms which use few shots for each node and achieve nearly optimal broadcasting time. To the best of our knowledge, the present paper is the first that addresses the issue of deterministic $k$-shot broadcasting in general radio networks of unknown topology.

\paragraph{Our contribution.}
In this paper we show (a) a lower bound of $\Omega(n^2 / k)$ on the broadcasting time of any oblivious $k$-shot broadcasting algorithm and (b) an oblivious broadcasting protocol that achieves a matching upper bound, namely $O(n^2 / k)$, for every $k \le \sqrt{n}$ and an upper bound of $O(n^{3/2})$ for every $k > \sqrt{n}$. This bound implies the following tradeoff between time complexity and the number of maximum transmissions per node:

$$ \#\verb|shots| \: \times \: \#\verb|steps|  = \begin{cases} 
\Theta(n^2) & \text{for } k \leq \sqrt{n},\\ 
\Omega(n^2) &  \text{for } k > \sqrt{n}. 
\end{cases} $$

In order to prove the lower bound, we develop a technique which, given an oblivious protocol, constructs a graph (a \emph{collision} graph) which succeeds in blocking the progress of broadcasting. It should also be noted that the lower bound holds even in the case of oblivious broadcasting in symmetric (undirected) networks. 

Our algorithm which matches the lower bound is based on the $O(n^{3/2})$-time algorithm of~\cite{CGOR00}. A very interesting consequence of our results, among others, is that we may impose a $\sqrt{n}$-shot restriction on the algorithm of~\cite{CGOR00} without affecting its performance.

Finally, we study the most general case of adaptive algorithms that are as strong as possible (in the sense that they can make use not only of their own history but also of the complete history of other nodes). Generalizing the lower bound approach of~\cite{BP97}, which refers to the case of any number of shots, we introduce a construction, which we call the \emph{transmission tree}, that allows us to obtain an $\Omega\left(n^{\frac{1+k}{k}}\right)$ lower bound for the $k$-shot case. We manage that, by specifying a correspondence between the height of the transmission trees and the time that the algorithm needs to complete broadcasting.

\section{Model and Preliminaries}

We model a radio network as a directed graph. This means that if between two nodes $u$ and $v$ there exists an edge $(u,v)$ but not the opposite edge $(v,u)$, then node $u$ can transmit to node $v$, but not vice versa. Furthermore, we assume that the nodes have unique labels from the set $V = \{1,2, \dotsc,n\}$, where $n$ is the number of nodes in the network. Initially, a node is aware only of its own label and whether it is the source node or not. This means that it has no knowledge, full or partial, about the topology of the underlying graph. We also assume that every node knows the size $n$ of the network. 

We consider protocols under the assumption that a node may transmit only after it has received the source message, i.e.\ there are no \emph{spontaneous transmissions}. Moreover, we assume that the nodes are not capable of \emph{detecting collisions}, that is, if an attempt to transmit to a node $v$ was unsuccessful, then $v$ is not able to sense it.

We say that a broadcasting algorithm (or protocol) completes broadcasting when all nodes of the network have received the source message. The \emph{running time} or \emph{broadcasting time} of an algorithm is defined as the worst-case number of steps needed to complete broadcasting over all possible network configurations (with the same number of nodes).(?)

We now define the notion of \emph{oblivious $k$-shot protocols}. As mentioned earlier, a protocol is oblivious if nodes do not take into account any information that may be gained during the execution of the protocol. Formally, an oblivious protocol can be succintly described as a sequence of \emph{transmission sets}, which are subsets of the node set $V$. Once a node receives the message at step $t$, it wakes up and transmits at the first $k$ steps after $t$ in which it appears in the transmission set.(?unrestricted broadcasting) This model captures an important class of broadcasting algorithms, since most known algorithms for deterministic broadcasting in networks with unknown topology fall into this class. 

Finally, we introduce some useful notation. We refer to a path graph $S$ as a \emph{chain} and denote by $V(S)$ the set of its nodes. For simplicity, we denote by $|S|$ the number of nodes in $S$. We also say that a graph $G$ \textit{starts} with a chain $S$ when $S$ is a subgraph of $G$, no node in $S$ but the last is connected to nodes in $V \setminus V(S)$ and the source is the first node of the chain. We define the concatenation of two chains $S_1$ and $S_2$, denoted by $S_1 \circ S_2$, as the graph consisting of $S_1$ and $S_2$ with the last node of $S_1$ connected to the first node of $S_2$, and no other edge between $S_1$ and $S_2$. We will also denote by $S \circ w$ the concatenation of chain $S$ with the chain consisting of a single node $w$.

\section{A Lower Bound for Oblivious $k$-shot Broadcasting} \label{sec:LBk}

In this section we prove an $\Omega(n^2/k)$ lower bound for any oblivious $k$-shot broadcasting protocol. We first need to introduce some more notation. 

We denote by $\mathtt{shots}(v,T)$ the minimum between $k$ and the number of times $v$ appears in a transmission set after step $T$. Note that  $\mathtt{shots}(v,T)$ is the number of times node $v$ will transmit if it receives the message at step $T$. Let also $t_i(v,T)$, where $i \leq \mathtt{shots}(v,T)$, be the step where node $v$ appears for the $i$-th time in a transmission set after step $T$. Moreover, we define
$$ t_{\leq i}(v,T)=\begin{cases} 
t_i(v,T) & \text{if } i \leq \mathtt{shots}(v,T),\\ 
t_{\mathtt{shots}(v,T)}(v,T) & \text{otherwise}. 
\end{cases} 
$$

We say that a sequence of nodes $S = \langle v_1, \dotsc, v_{|S|}\rangle$ \emph{occurs} in a schedule $\mathcal{S}$ if there is a subsequence $\mathcal{S}' = \langle \mathcal{T}_1, \dotsc, \mathcal{T}_{|S|}\rangle$ of transmission sets in $\mathcal{S}$ such that for all $i=1,\dotsc, |S|$, it holds that $v_i \in \mathcal{T}_i$. The first occurrence of sequence $S$ after some step $T$ is defined in a similar way, where in addition we ask for $\mathcal{T}_1$ to appear after step $T$ and for $\mathcal{T}_{|S|}$ to appear as early as possible in the schedule. We denote by $t_1(S,T)$ the step where $ \mathcal{T}_{|S|}$ appears.  

Now, let us consider an oblivious $k$-shot broadcasting protocol $\mathcal{P}$. 

\begin{lem} \label{confl_sub}
Consider a sequence $S$ and let $T = t_1(S,0)$. Consider also any set $Q  \subseteq V \setminus V(S)$. Then, there exists a node $w \in Q$ such that $$t_{\leq k}(w,T) \geq T+|Q|-1$$
\end{lem}

\begin{proof}
We first define a bipartite graph $G = (A,B,E)$ as follows. The upper set $A$ corresponds to the nodes in the set $Q' = Q \setminus \{v_t\}$, where $v_t$ is an arbitrary node in $Q$. Let $L = \max_{u \in Q'} t_{\leq k}(u,T)$, that is, the last step in which some node from $Q'$ transmits if it receives the message at step $T$. We set node $w$ to be the node that maximizes $L$. 
The lower set $B$ corresponds to steps from $\{ T+1, \dotsc, L \}$ in which some node from $Q'$ transmits. If node $u$ transmits at step $t$, we add an edge between $u \in A$ and $t \in B$. 

We say that an induced subgraph $H$ of $G$ is \textit{conflicting} if for any node $w \in V(H)$ the following two properties hold:
\begin{enumerate}
\item If $w \in A$, then all neighbors of $w$ in $G$ also belong to $V(H)$
\item If $w \in B$, then $\deg_H(w)> 1$ 
\end{enumerate}
  
Note that in $H$, no vertex of $B$ has only one neighbour, which means that all transmission sets included in $H$ contain at least two nodes. 

Let us now state and prove the following property: \textit{$G$ contains no conflicting subgraphs}. 

Indeed, suppose that $G$ has some conflicting subgraph $H = (A',B',E')$. Consider the graph $G_t$ with the following topology: graph $G_t$ starts with the chain corresponding to $S$ and the last node $v_S$ of $S$ is connected to all nodes in $A'$. Moreover, $v_t$  is connected to every node in $A'$ and has no other neighbors. At step $T$, $v_S$ transmits and all nodes in $A'$ get the message. However, nodes in $A'$ transmit only at steps in $B'$, according to the first property of $H$. Since every step in $B'$ has at least two neighboring nodes in $A'$, the corresponding transmission set contains at least two nodes possessing the source message and therefore  a conflict occurs at every such step. Moreover, by the end of step $L$, every node in $A'$ either has transmitted $k$ times or has no more transmissions available. Thus, $v_t$ never gets the message, which is a contradiction.    
  
Based on this property, it is easy to see that $G$ has at least one vertex $u \in B$  with $deg_G(u)=1$ (otherwise $G$ would be a conflicting subgraph itself). Suppose now that we remove $u$ along with its \textit{only} neighbor to obtain graph $G'$. Notice that $G'$ is an induced subgraph of $G$ and for any node $u \in A \cap G'$, all its neighbors belong to $V(G')$. Consequently, $G'$ cannot be a conflicting subgraph, thus there exists some vertex $u' \in B$ such that $deg_{G'}(u')=1$. This process may continue $|Q|-1$ times (since each time we remove at most one node from $A$), until all nodes in $A$ are chosen. Intuitively, this process maps each node to a unique step in $B$. Thus, $|B| \geq |Q|-1$, therefore $L \geq T+|Q|-1$.          
\end{proof}

\begin{lem} \label{sub_k}
Let us fix the first occurrence of sequence $S$ at step $T$. Then, there exists a sequence $R$ of length at most $k$ such that $t(R,T) \geq T+n-|S|-k$.
\end{lem}

\begin{proof}
Let $Q= V \setminus V(S)$. We first construct a set $W = \{w_1, w_2, \dotsc, w_k \} \subseteq Q$ of size $k$ as follows. We apply lemma \ref{confl_sub} to $Q$; thus, there exists a node $w_1$ such that $t_{\leq k}(w_1,T) \geq T+|Q|-1$. Generally, for any $i \leq k$, we apply lemma \ref{confl_sub} to the set $Q \backslash \{w_1, \dotsc, w_{i-1}\}$ to obtain a node $w_i$ such that $t_{\leq k}(w_i,T) \geq T+|Q|-i$. Thus, the set $W$ we obtain has the following property: for any $w \in W$, it holds that $t_{\leq k}(w,T) \geq T+|Q|-k$.

Next, we use the nodes from $W$ to construct a sequence $R = \langle r_1 r_2 \dotsc r_k \rangle$ of length $k$. We will assume that any symbol of $R$ may be the empty symbol, which we denote by $\varepsilon$. The construction will be as follows: in order to compute $r_i$, we calculate the value $M_i = \max_{v \in W \backslash \{ r_1,\dotsc,r_{i-1}\}} \{ t_{\leq i}(v,T) \}$. If $M_i \leq \max_{j<i}\{M_j\}$, then we set $r_i = \varepsilon$. Otherwise, $r_i$ is the node which maximizes $M_i$. Intuitively, we find the node which transmits last for the $i$-th time. 

We will show that $t(R,T) \geq \max_i \{M_i\}$. In order to prove this, we will prove by induction that $t(\langle r_1r_2\dotsc r_i\rangle,T) = \max_{j \leq i} \{M_j\}$. Here, let us note that, since we construct a chain of nodes, nodes other than the first node will receive the message after step $T$ and thus their transmissions will not coincide with the transmissions that would occur in the case that all nodes receive the message at $T$. However, the assumption that all nodes behave as if they receive the message at $T$ only favors the progress of broadcasting and thus is a lower bound of the time needed for the message to reach the last node of the chain. We will therefore use this assumption to simplify the proof.

For the induction base, we note that $r_1$ transmits first at step $M_1$, thus the proposition trivially holds. Now, let us examine the sequence $\langle r_1r_2\dotsc r_{i+1} \rangle$. By the induction hypothesis, we know that $t(\langle r_1 \dotsc r_i\rangle,T) \geq \max_{j \leq i} \{M_j\}$. 
We now distinguish two cases.
\begin{itemize}
\item $r_{i+1} = \varepsilon$: Then, we have $t(\langle r_1 \dotsc r_{i+1} \rangle,T) = t(\langle r_1 \dotsc r_i\rangle,T)$. Moreover, due to the construction, it holds that  $M_{i+1} \leq \max_{j<i+1}\{M_j\}$. Thus, $\max_{j \leq i} \{M_j\} = \max_{j \leq i+1} \{M_j\} $ and, using the induction hypothesis, the proposition holds.  
\item $r_{i+1} \neq \varepsilon$: In this case, it holds that $t_{i+1}(r_{i+1},T)=M_{i+1}$ and that $M_{i+1} > \max_{j<i+1}\{M_j\}$.  Thus, $\max_{j\leq i+1}\{M_j\} = M_{i+1}$. Furthermore, the previous $i$ occurrences of $r_{i+1}$ are not after step $\max_{j \leq i}\{M_j\}$ by construction. Consequently, $t(\langle r_1 \dotsc r_{i+1}\rangle,T) = M_{i+1} = \max_{j\leq i+1}\{M_j\}$. 

\end{itemize}

Now, let us consider the last node $r_j$ of $R$ such that $r_j \neq \varepsilon$. Clearly, it holds that $\forall i \neq j: M_j \geq M_i$ and that $t(R,T) = \max_i \{M_i\} = M_j = t_{\leq j}(r_j,T)$. In the case where $j=k$, we have that $t_{\leq j}(r_j,T) = t_{\leq k}(r_k,T)$. Otherwise, since $r_i = \varepsilon$ for $i>j$, any node in $\{w_{j+1}, \dotsc, w_k \}$ does not occur after step $t_{\leq j}(r_j,T)$, thus $t_{\leq j}(r_j,T) \geq t_{\leq k}(p,T)$ for any $p \in \{ w_{i+1}, \dotsc, w_{k}\}$. 
 
In any case, there exists a node $q \in W$ such that $t(R,T)  \geq t_{\leq k}(q,T)$. However, $t_{\leq k}(q,T) \geq T+|Q|-k$ using the property of set $W$. Consequently, $t(R,T) \geq T+|Q|-k = T+n-|s|-k$.
\end{proof}

The proof of the lower bound uses repeatedly lemma \ref{sub_k}. 

\begin{thm} \label{kshot}
For any oblivious $k$-shot broadcasting protocol $\mathcal{P}$, there exists a graph $S_{\mathcal{P}}$ where $\mathcal{P}$ needs $\Omega(n^2/k)$ steps to complete broadcasting.  
\end{thm}

\begin{proof}
Consider the source $a_0$. Since the sequence $\langle a_0 \rangle$ first occurs at step $1$, from lemma \ref{sub_k}, there exists a sequence $a_1$ of length at most $k$ such that $T_1 = t(\langle a_0\rangle \circ a_1) \geq 1+n-k$. 

Now, we may apply lemma \ref{sub_k} to $a' = \langle a_0\rangle \circ a_1$ so as to find a sequence $a_3$, where $T_2 = t(\langle a_0 \rangle \circ a_1 \circ a_2) \geq T_1+n-2k$.  We can continue this process until we have constructed a sequence of length $k$ or less. Thus, we can construct a sequence $S$ of nodes which occurs for the first time at step

$$1+ \sum_{j=1}^{\lfloor \frac{n}{k} \rfloor} (n-jk) + \Omega(1) $$    
The chain $S_{\mathcal{P}}$ corresponding to sequence $S$ is our graph. 
\end{proof}

\section{An Oblivious Algorithm for $k$-shot Broadcasting}

We will present an oblivious algorithm (\textsc{Oblivious k-Shot}) which is an adaptation of the algorithm presented in  \cite{CGOR00} and performs optimal $k$-shot broadcasting in time $O(n^2/k)$ for any $k \leq \sqrt{n}$. For $k \geq \sqrt{n}$, the algorithm completes broadcasting in $O(n^{3/2})$ steps, which matches the time performance of the algorithm in \cite{CGOR00}.

Let $p$ be the smallest prime greater than or equal to $\lceil \sqrt{n} \rceil$. We map a node with label $i$ to the point $\langle i \; \text{div} \; p, i \bmod p\rangle$. A \textit{line} $L_{a,b}$ with \textit{direction} $a$ ($a=0, \dotsc, p$) and \textit{offset} $b$ ($b=0, \dotsc, p-1$) is defined as the following set of points:

$$ L_{a,b}=\begin{cases} 
\{ \langle x,y \rangle : x \equiv b \pmod p \}& \text{if $a = p$},\\ 
\{ \langle x,y \rangle : y \equiv a \cdot x + b  \pmod p \} & \text{else}. 
\end{cases} 
$$

It is easy to observe that the sets defined have the following useful properties, which will be crucial in analyzing the running time of the algorithm.
\begin{itemize}
  \item The total number of distinct lines is $p \cdot (p+1)$.
  \item Each node belongs to $p+1$ lines, one in each direction.
  \item There are $p$ disjoint lines in each direction.
  \item Two lines of different directions have exactly one common node.
  \item For any two different nodes, there is exactly one line that contains both of them.  
\end{itemize}

The algorithm multiplexes two different procedures, the classic \textsc{Round-Robin} procedure where the nodes transmit one after the other alone, and the \textsc{Line-Transmit} procedure, where lines are used as transmission sets.

\begin{procedure}[H] 
\SetLine 
\Repeat{}{
\For{ $v=1,2, \dotsc, p^2$}
{  node $v$ transmits
}}
\caption{\textsc{Round-Robin}()} 
\end{procedure} 

\begin{procedure}[H] 
\SetLine 
\dontprintsemicolon
\Repeat{}{
\For{$a=0, \dotsc, p$}
{ \textsc{Stage} \;
\For{$b=0, \dotsc, p-1$} { 
	all nodes in $L_{a,b}$ transmit
}}}
\caption{\textsc{Line-Transmit}()} 
\end{procedure} 

We define $K = \lceil \frac{p}{k-2}\rceil$ and the procedures are multiplexed such that a step of the \textsc{Line-Transmit} is followed by $K$ steps of the \textsc{Round-Robin} procedure.

\begin{algorithm}[H] 
\SetLine 
\Repeat{}
{ perform a \textsc{Line-Transmit} step \;
  perform $K$ \textsc{Round-Robin} steps \;
}
\caption{\textsc{Oblivious k-Shot}} 
\end{algorithm}

Before we analyze the time complexity of the algorithm, we need to examine whether the \textsc{Oblivious K-Shot} algorithm is valid. In order to show this, we will prove that every node transmits at least one time alone.
\begin{lem}[\textsc{Validity}] \label{alone}
In \textsc{Oblivious k-Shot}, every node transmits at least once as the only node in the network. 
\end{lem}
\begin{proof}
Assume that node $v$ receives the message at some step $t$. We will show that $v$ transmits at least once at a step of the \textsc{Round-Robin} procedure. In order to obtain a contradiction, let us assume that $v$ transmits all $k$ times during the \textsc{Line-Transmit} procedure. Since each node transmits exactly once during a stage, this means that $v$ transmits at $k$ consecutive stages. 

We will now compute the number of \textsc{Round-Robin} steps during these $k$ stages. Note that each stage has $p$ line transmissions and after each line transmission occur $K$ \textsc{Round-Robin} steps. Thus, we have at least $( p \cdot (k-2)+1 ) \cdot K > p \cdot (k-2) \cdot  \lceil \frac{p}{k-2}\rceil > p^2$ \textsc{Round-Robin} steps. Clearly, this means that $v$ would appear at least once at a \textsc{Round-Robin} step during this period, which yields a contradiction.  
\end{proof}
 
 Let us now examine the time complexity of this broadcasting protocol. We say that we make \textit{progress} each time a node receives the message or when all its neighbors have received the message. Achieving progress $2n-1$ means that every node has received the message and thus broadcasting is complete. 

We next prove a useful lemma. We denote by $r$ a constant $r < k$ which we will fix later so as to optimize the analysis.

\begin{lem} \label{prog}
Let $F$ be the active nodes at the beginning of stage $t$. If $|F| \leq k/r$, then the average progress per stage during the next $r \cdot |F|$ stages is constant. 
\end{lem}

\begin{proof}
Let us consider the consecutive stages $T = \{ t, t+1, \dotsc, t+ r \cdot |F|-1\}$. If every node in $F$ broadcasts during $T$ at least once as the only node, then the progress is at least $|F|$ and thus the average progress per stage is at least $\frac{|F|}{r \cdot |F|}=\frac{1}{r}$. 

Otherwise, there exists some node $v \in F$  such that whenever $v$ broadcasts in $T$, then some other node broadcasts simultaneously and causes a collision. However, $v$ broadcasts only $k$ times during the odd-numbered steps and it is possible that $v$ has already transmitted before $t$. We will examine the behavior of node $v$ in a window of $k-1 < \sqrt{n} < p+1$ stages where $v$ transmits. Let us assume that $v$ has already transmitted $i$ times before $t$. Since $v$ is an active node, this means that there exists a neighboring node $w$ that has not received the message. Thus, in each of the previous $i$ transmissions, a node neighboring to $w$ had transmitted simultaneously. Moreover, these nodes are all different, since any two nodes broadcast simultaneously only once during a window of less than $p+1$ stages. Note also that the $i$ nodes will be active at the beginning of $T$, since they have a neighboring node ($w$) which has not received the message yet. From the hypothesis that $|F| \leq k/r$, we obtain that $i < k/r$. This means that $v$ transmits during $T$ at least $(k-k/r)$ times. Furthermore, these $i$ nodes will not transmit together with $v$ in the next $k-i-1$ transmissions.

Using a similar argument, and since a collision occurs at each of these $(k-k/r)$ steps, there will be at least $(k-k/r)$ different nodes outside $|F|$ which transmit together with $v$. This implies a progress of at least $(k-k/r)$ and an average progress per stage of at least $\frac{k-k/r}{r \cdot |F|} \geq \frac{k-k/r}{k} = 1-\frac{1}{r}$.

Thus, the average progress per stage is at least $\min \{\frac{1}{r}, 1-\frac{1}{r} \}$, which is a constant.
\end{proof}

It is easy to see that the average progress is maximized when $r=2$. Then, we have that if $|F| \leq k/2$, the average progress per stage is at least $1/2$.

\begin{thm}
\textsc{Oblivious k-Shot} completes $k$-shot broadcasting in $O(n^2/k)$ steps for $k \leq \sqrt{n}$.
\end{thm}

\begin{proof}
We need to calculate the time in order to make progress $2n-1$. Clearly, if at any stage the number of active nodes is more than $k/2$, then the \textsc{Round-Robin} procedure guarantees that the progress over the next $O(n)$ steps is at least $k/2$. Thus, the average progress per step is $\Omega(k/n)$. Otherwise, the number of active nodes will be at most $k/2$ and thus the progress per stage will be $\Omega(1)$ by using lemma \ref{prog} with $r=2$. Since each stage has $p \cdot (K+1) = O(n/k)$ steps, the average progress per step is again $\Omega(k/n)$. Consequently, we have that the total time complexity will be $O(n^2/k)$. 
\end{proof}

Finally, it is easy to observe that using \textsc{Oblivious K-Shot} with $k=\sqrt{n}$, we obtain an $O(n^{3/2})$ $k$-shot broadcasting algorithm for any $k > \sqrt{n}$.


\section{Adaptive broadcasting protocols}
For a formal definition of an adaptive broadcasting protocol we will use a slight generalization of the model proposed by Kowalski and Pelc~\cite{KP04}. In our model for adaptive protocols, we allow a node to transmit a message to its neighborhood even before it receives the source message. An algorithm may use this kind of transmission for topology knowledge exchange which may influence actions in later rounds.
 We denote by $H_t(v)$ the view of node $v$ until the end of step $t$, i.e., a complete description of all the messages received (along with the corresponding sender's id) and send by  $v$ during each round $1,\ldots,t$. We will use the notion \emph{incoming view} for the description of the incoming messages. 

A broadcasting protocol can now be defined by a function $\pi(v, t, H_{t-1}(v))$, which takes values in the set \{\texttt{receive, transmit}\}. The function decides whether node $v$ with view $H_{t-1}(v)$ acts as a receiver (\texttt{receive}) or as a transmitter (\texttt{{transmit}}) at step $t$. If $v$ acts as a transmitter in step $t$, it sends its entire view until step $t-1$ along with its id, i.e., the message $(v, H_{t-1}(v))$. Note that the maximum information exchange occurs when transmitters send their entire view, since the receiver can always deduce any information from the received view. For completeness we assume that in the initialization phase of the protocol (\emph{step 0}) each player  $v\in V\setminus \{s\}$ has the view $H_0(v)=(\emptyset, \bot)$ which actually represents the lack of an initial input value. The source node $s$ has the initial view $H_0(s)=(\emptyset, m)$, where $m$ is the initial value which will be propagated by $s$.

\subsection{Broadcasting protocols and transmission Trees  }

Generalizing the lower bound approach of~\cite{BP97} for  any number of shots we introduce a construction, which we call the \emph{transmission tree}, that allows us to obtain a lower bound for the $k$-shot case. We consider the most general case of adaptive protocols, as defined above and for any such protocol $\pi $ we construct a network in which the delay of completing broadcasting with $\pi$ is significantly increased. We use the transmission tree tool to maximize the delay in the intermediate stages of achieving Broadcast.

\subsubsection*{Family of networks}
The class of radio networks $\calg$ we will use for our lower bound argument are graphs of a certain topology, namely the $n$ nodes of each such graph can be partitioned in $l$ layers; The first layer contains only the source node  $s$ and the next $l-2$ contain 2 nodes each and the last layer contains the rest of the nodes (1 or 2) to complete the partition. Moreover each node $v$ in layer $i$ is connected, with a directed edge $(v,w)$, to each node $w$ of layer $i+1$ and no other connections exist.

More concretely, For an $n$-node graph $G=(V,E)\in \calg$, $V$ can be partitioned in $l=\lfloor n/2\rfloor+1$ layers $L_1,\ldots,L_l$  s.t. $L_1=\{s\}$,  $|L_2|=\ldots=|L_{l-1}|=2$ and $L_l=V\setminus \cup_{i=1}^{l-1}L_i$. Moreover, $E=\{(w,v)\in L_i\times L_j \mid i,j\in\{1,\ldots l\}, j-i=1\}$. Having specified the topology of the family $\calg$ the different members of the family differ in the number of nodes and the assignment of the id's. The general topology structure of family $\calg$ is depicted in Figure~\ref{fig:netopology}.    

\begin{figure}[h]
\centering
\includegraphics[width=0.4\linewidth]{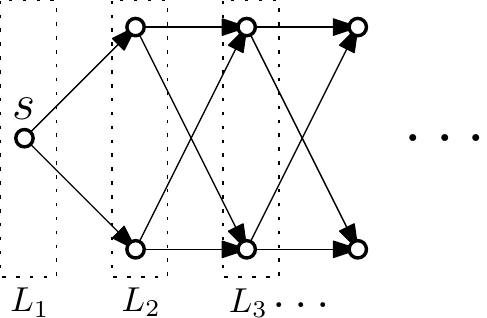}
\caption{Family of graphs $\calg$.}
\label{fig:netopology}
\end{figure}

\subsubsection*{Designing a ``bad'' graph}

Let us consider any deterministic $k$-shot broadcasting protocol $\pi$ which completes broadcasting in any graph with n nodes. We will construct a graph $G_\pi\in \calg$, by assigning ids to nodes, such that broadcasting is significantly slowed down. We will gradually construct graph $G_\pi$ by using the graph families $\calg_i$ as described in the following.

For an arbitrary assignment $ID_i$ of ids in the first $i$ ($i\in \{1,\ldots, l-1\}$) layers of family $\calg$ graphs, we define family $\calg_i\subseteq \calg$ of graphs with $n$ nodes that have the assignment $ID_i$ in their first $i$ layers and thus differ only in the next $l-i$ layers. Let $S$ be the set of the assigned ids and $A=V\setminus S$. Consider the execution of protocol $\pi$ in any $G\in \calg_i $.  Observe that in all graphs $G\in \calg_i$  nodes in layer $L_{i+1}$ receive from $L_i$ the same incoming view $H_j$, for every round $j$ of the execution. Incoming view sequence $(H_j)=(H_1, H_2, \ldots)$ can be determined through protocol $\pi$ given that the topology and the id's of the first $i$ layers are known. Observe that according to the view definition the term $H_j$ contains all the information contained in terms $H_1, \ldots, H_{j-1}$. We make a distinction between the incoming and the full view because in the following we will consider the actions that different nodes take given that they receive  the same view.

 Since $L_i$ contains all the incoming neighbors of nodes in $L_{i+1}$, and there is no directed path from nodes in $L_{i+1}$ to nodes in $L_i$, $(H_j)$ is guaranteed  to capture the whole (incoming) view of nodes in $L_{i+1}$ and thus their actions (\texttt{transmit, receive}) can be determined for every round of the execution. 
 To determine if node $v\in A$  transmits in round $t$ under incoming view $H_{t-1}$ (being in level $i+1$) one should simulate the execution of protocol $\pi$, where $v$ receives the view $H_{t-1}$
and construct $v$'s view $H_{t-1}(v)$ which may additionally contain messages sent by $v$~\footnote{Obviously the messages sent by $v$ don't have any effect on the incoming view due to the topology (no directed path exists from nodes in $L_{i+1}$ to nodes in $L_{i}$).}. Working this way we can define the round $t$ transmitting/sending nodes which belong to a set $A$ with 
$$S_t^A=\{v\in A \mid \pi(v,t,H_{t-1}(v))=`` \texttt{transmit}"\}$$
  and  the round $t$ receiving nodes with 
  $$R_t^A=\{v\in V \mid \pi(v,t,H_{t-1}(v))=``\texttt{receive}"\}$$.

Given a family of graphs $\calg_i$ and a protocol $\pi$ we will construct the family of graphs $\calg_{i+1}\subseteq \calg_i$ in which the delay of transmitting the message from $L_{i+1}$ to $L_{i+2}$ maximizes. Our approach is summarized to the ``worst" choice of the two ids of layer $L_{i+1}$, such that either by collision or non-transmission the message fails to transmit to $L_{i+2}$ for the maximum number of steps.

\subsubsection*{Transmission trees}

\begin{definition}
 A \emph{transmission tree} $T(A,\pi,(H_j))$ corresponding to broadcasting protocol $\pi$, incoming view sequence $(H_i)$ and node set $A$  is  a binary tree. The id of the root is the node set $A$ and every child's id is a subset of its father's id. The ids of the children form a partition of the father's id and all the leaves are singletons. For a node $P$ at depth $t-1$, its left child's id is $R_t^P$ whereas the right child of $P$ is the set $S_t^P$, i.e.,  the nodes in $P$ which are receiving (don't transmit), respectively transmitting, in round $t$ given that their incoming view is $(H_j)$.
 \end{definition}
 
 
Observe that given a family $\calg_i $ (which implies the function $ID_i$) and a protocol $\pi$ we can uniquely define, as previously described, the set of unassigned ids $A$ and the incoming view sequence $(H_j)$ that all nodes in layer $L_{i+1}$ receive in $\calg_i$ during the execution of $\pi$. We therefore use the notation $T(\pi, \calg_i)$ to refer to the transmission tree $T(A,\pi,(H_j))$ where $A$ and $(H_j)$ can be deduced by $\pi, \calg_i$ as described above. 
 
%
 
 More importantly, given a family $\calg_i$, any broadcasting protocol $\pi$ uniquely defines a transmission tree $T(\pi, \calg_i)$ which describes the actions that players in $A$ take, under the reception of the corresponding view $(H_j)$, executing protocol $\pi$.  
  The only non-trivial point in this correspondence is why the leaves, of the corresponding tree of any protocol $\pi$, have to be singletons. The reason for this is that if there was a leaf $P$ with $|P|\geq 2$ then this protocol would never achieve broadcasting in the family $\calg_{i+1}$ with $L_{i+1}\subseteq P$. 
  
\begin{definition} A \emph{$k$-shot transmission tree} is a transmission tree in which each branch contains at most $k$ right children.
\end{definition}
 This in fact ensures that each node will transmit in at most $k$ steps as desired for a $k$-shot protocol.

\begin{theorem}\label{theorem:1}

Given a broadcasting protocol $\pi$ and a family $\calg_i$, 
there exists  $\calg_{i+1}\subseteq \calg_i$ s.t. the relay of the message from layer $L_{i+1}$ to layer $L_{i+2}$ in all graphs of $\calg_{i+1}$ will be delayed as much as $height(T(\pi, \calg_i))$.

\end{theorem}

\begin{proof}
To create the family $\calg_{i+1}$ we only have to choose the two ids $v_1,v_2$ that will be assigned to layer $L_{i+1}$. The relay of the message from $L_{i+1}$ to $L_{i+2}$ will happen in the first round that only one of $v_1,v_2$ will transmit, in a different case either a collision or a non-transmission phase will occur. The definition of the transmission tree implies that there is a tree-node $P$ with $|P|\geq 2$ and $depth(P)=height(T)-1$. Choosing nodes $L_{i+1}=\{v_1,v_2\}\subseteq P$, the first round where one of them transmits alone will be round $t>depth(P)=height(T)-1$, therefore transmission to $L_{i+2}$ will be delayed until round $height(T(\pi, \calg_i))$.

\end{proof}

\begin{theorem}\label{theorem:gi}
The minimum height of a $k$-shot transmission tree for family $\calg_i$ with unassigned ids $A$ and $|A|=a$ over all $k$-shot broadcasting protocols $\calb$ is 
$$\min_{\pi\in  \calb}height\left( T(\pi,\calg_i)\right)=\Omega(a^{\frac{1}{k}})$$
\end{theorem}

\begin{proof}
Wlog we can assume only protocols corresponding to transmission trees where every internal node has a  right child. The reason for that is that if a protocol $\pi$ corresponds to a tree in which a node $v$ only has a left (non-transmitting) child $w$, then deleting this edge along with the left child $w$ and connecting the children of $w$ to $v$ will result to a transmission tree of  non-greater height and thus a protocol which achieves a non-slower relay. 

For the case of $k=1$ one can observe that each right  child $P$ will contain only one node ($|P|=1$) and will be a leaf. This is obvious from the definition of the $k$-shot transmission tree; since each branch contains at most $k=1$ right child and all the leafs are singletons, each right child must be a singleton-leaf. Therefore the minimum height tree will result if the root and every left child has a right child leaf. Subsequently for the case of $k=1$, 
$$\min_{\pi\in \calb} T(\pi,\calg_i)=a-1=\Omega(a)$$
Therefore the theorem holds for $k=1$.

Assume that the claim holds for $k=i-1$, then we will prove that it holds for $k=i$. First consider an $i$-shot transmission tree $T$ and its leftmost branch $LB$ including the root. Observe that each right child $P$ of a node in $LB$ is actually a root of an $(i-1)$-shot transmission tree since all nodes in $P$ have only $i-1$ shots left. By the induction hypothesis we know that the the minimum height of every such tree is $\Omega\left(|P|^\frac{1}{i-1}\right)$.

For any $i$-shot transmission tree $T$ there are two cases (two types of transmission trees):

\begin{enumerate}
\item Every right child $P$ of nodes in $LB$  has cardinality $|P|=O\left(a^\frac{i-1}{i}\right)$. In this case, the length of $LB$ in this tree will be of order 

$$|LB|=\frac{a}{O\left(a^\frac{i-1}{i}\right)}=\Omega \left( \frac{a}{a^\frac{i-1}{i}}\right)=\Omega(a^{1-\frac{i-1}{i}})=\Omega(a^{1/i}) $$ 

 since for every pair of $LB$ nodes  $P_r, P_{r+1}$ of depth $r$ and $r+1$ respectively, it holds that $|P_{r+1}|=|P_r|-O(a^{\frac{i-1}{i}})$. 
 
Moreover it holds that $height(T)\geq |LB|= \Omega(a^\frac{1}{i})$ and therefore $$height(T)=\Omega(a^\frac{1}{i})$$.

\item There exists a right child $P$ of nodes in $LB$  with cardinality $$|P|\neq O(a^\frac{i-1}{i})\Leftrightarrow|P|=\omega(a^\frac{i-1}{i})$$

By the induction hypothesis every such tree $T_P$ with root $P$ will have a minimum height of order,

$$height(T_P)=\Omega\left( \left(\omega(a^{(i-1)/i})\right)^\frac{1}{i-1}\right)= \Omega \left(  a^{\frac{i-1}{i}\cdot\frac{1}{i-1}}\right)=\Omega\left(a^\frac{1}{i}\right) $$

Moreover it holds that $height(T)\geq height(T_P)= \Omega(a^\frac{1}{i})$ and therefore, in this case also it holds that,
 $$height(T)=\Omega(a^\frac{1}{i})$$.

\end{enumerate}
 
Therefore the minimum height of any $k$-shot transmission tree is 

$$\min_{\pi\in  \calb}height\left( T(\pi,\calg_i)\right)=\Omega(a^{\frac{1}{k}})$$

\end{proof}

\begin{theorem}
For any $k$-shot adaptive broadcasting protocol $\pi$, there is a $n$-node graph $G\in \calg$ where $\pi$ needs $\Omega(n^{\frac{1+k}{k}})$ rounds to achieve Broadcast.
\end{theorem}

\begin{proof}
Repeatedly applying theorems~\ref{theorem:1},\ref{theorem:gi} for $i=1, \ldots \lfloor n/2 \rfloor$ we construct a graph in which $\pi$ will achieve broadcasting in time asymptotically greater or equal than
$$S_1=(n-1)^{1/k}+(n-3)^{1/k}+\cdots+2^{1/k}$$
in the case where $n$ is odd and 
$$S_2=(n-1)^{1/k}+(n-3)^{1/k}+\cdots+3^{1/k}$$
when $n$ is even.

Observe that in the case where $n$ is odd the sum $S_1$ is comprised by $\lfloor n/2\rfloor$ terms and the half of these terms are lower than the $(n/2)^\frac{1}{k}$. Hence it holds that, 
$$S_1\ge \left\lfloor \frac{n}{4} \right\rfloor \cdot \left(\frac{n}{2} \right)^\frac{1}{k}>\left(  \frac{n}{4}-1 \right)  \cdot \left(\frac{n}{2} \right)^\frac{1}{k} \ge \frac{1}{8} (n-4)n^\frac{1}{k}\Rightarrow S_1=\Omega\left(n^{\frac{1+k}{k}}\right)$$

Where the last inequality holds because $2^\frac{1}{k}\geq 2, k\in \mathbb{N}$. Using similar arguments we can show that $S_2=\Omega\left(n^{\frac{1+k}{k}}\right)$.

%
%
%


\end{proof}

In the conference version of this paper~\cite{KP11}, only the case of adaptive algorithms for the 1-shot case was studied. A different approach was followed to prove a lower bound for this specific case and is deferred in the appendix because it might be interesting per se. The case however is captured in the result of this section for $k=1$.

\section{Conclusions}

In this paper, we initiate the study of deterministic $k$-shot broadcasting in radio networks with unknown topology. For the oblivious model, we manage to show an exact energy-time tradeoff for values of $k \leq \sqrt{n}$. It remains an open question to match the lower bound of $\Omega(n^2/k)$ for all values of $k$.

For the most general case of adaptive $k$-shot protocols, we manage to obtain a lower bound on the broadcasting time which differs from the bounds of the oblivious case. This means that either the bound is not tight or the information exchange between the nodes can actually guarantee a lower termination time in broadcasting. This may happen due to exchange of topology knowledge that the nodes obtain throughout such a protocol. The study of the bound's tightness as well as the devise of adaptive $k$-shot protocols constitute interesting research directions.

\bibliographystyle{plain}
\bibliography{kshot}

\newpage

\appendix

\begin{center}
 \textbf{\huge Appendix}
\end{center}

\section{Alternative proof for the adaptive 1-shot case}

%
%

Let us consider any deterministic $1$-shot broadcasting protocol $\mathcal{P}$ which completes broadcasting in any graph with $n$ nodes. We construct a graph $G_\mathcal{P}$ such that broadcasting is slowed down as much as possible. We start by considering the family $\mathcal{G}_0$ of all possible connected graphs with $n$ nodes. The construction proceeds by considering the steps of protocol $\mathcal{P}$, at each step refining the family of graphs. We will show that by the end of the construction, we are left with a graph in which $\mathcal{P}$ completes broadcasting in $\Omega(n^2)$ steps.  

We divide the construction into $n-3$ \textit{stages} and denote by $\mathcal{G}_i$ the family of graphs by the end of stage $i$. The construction is based on the following lemma.

\begin{lem} \label{slow_n}
Assume that the family $\mathcal{G}_{i-1}$ includes only graphs which start with a chain $S$ with $i$ nodes ($i \leq n-3$) and the last node $v_S$ of $S$ transmits no earlier than step $T$. Then, there exists a node $w \in \overline{S} = V \setminus V(S)$ such that the family $\mathcal{G}_{i}$ includes only graphs which start with chain $S \circ w$ and $w$ transmits no earlier than step $T+n-i$.
\end{lem}

\begin{proof}
For simplicity, we set $\mathcal{H}_0 = \mathcal{G}_{i-1}$. We also denote by $\mathcal{H}_0^{w} \subseteq \mathcal{H}_0$ the family of graphs in $\mathcal{H}_0$ where node $w$ is connected to $v_S$. Without loss of generality, let us assume that node $v_S$ transmits its history $H_{T-1}(v_S)$ at step $T$. In any graph $G \in \mathcal{H}_0^w$, node $w$ receives the same history $H_{T-1}(v_S)$ at step $T$ and has received only empty history before step $T$. Thus, for any graphs $G_a, G_b \in \mathcal{H}_0^w $, it holds that  $H_{T}(w,G_a) = H_{T}(w,G_b) = H_{T}(w)$. Since protocol $\mathcal{P}$ determines the action of any $w$ at step $T+1$ from $\pi(w,T+1,H_{T}(w))$, the action of any node $w$ is the same for any graph in $\mathcal{H}_0^w$ and we denote it by $\pi_{S}(w)$.

Assume that for $w_i, w_j \in \overline{S}$ it holds that  $ \pi_{S}(w_i) = \pi_{S}(w_j) = \verb|send|$. Then, consider any graph $G_{ij} \in \mathcal{H}_0^{w_i} \cap \mathcal{H}_0^{w_j}$ such that $w_i$ and $w_j$ are the only nodes connected to a node $v_t$. At step $T+1$, the nodes $w_i$ and $w_j$ transmit simultaneously, a collision occurs and thus $v_t$ never gets the message, a contradiction (see Figure~\ref{collision}). Thus, there exists at most one node $w_1 \in \overline{S}$ such that $\pi_{S}(w_1) = \verb|send|$. In this case, we refine the family of graphs to the family $\mathcal{H}_{1} = \mathcal{H}_{0} \setminus \mathcal{H}_0^{w_1}$ and set $\overline{S}_1 = \overline{S} \setminus \{w_1\}$. Otherwise, when no node from $\overline{S}$ transmits, we set $\mathcal{H}_{1} = \mathcal{H}_{0}$ and $\overline{S}_1 = \overline{S}$.

\begin{figure}[h]
  \begin{center}
   {\includegraphics[scale=0.6]{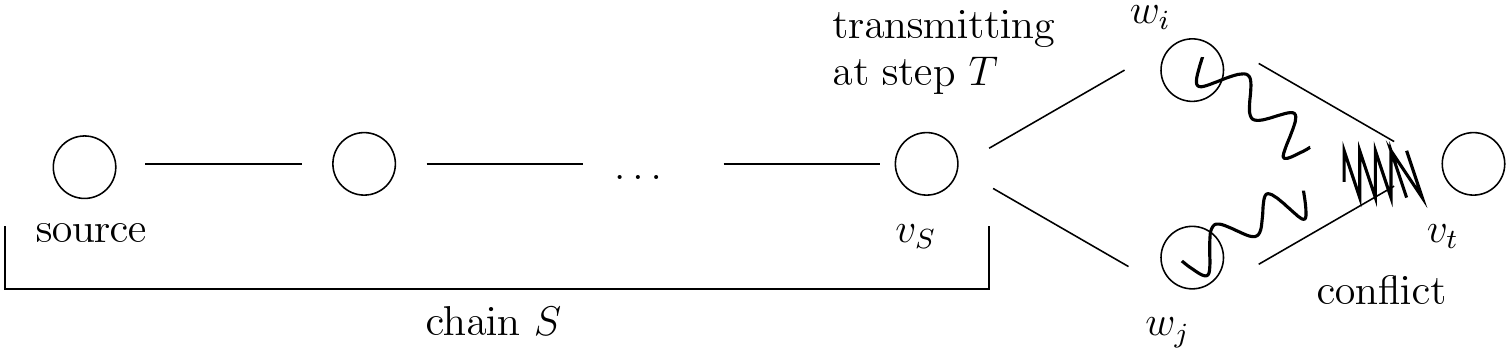}} 
   \caption{A collision graph: a conflict occurs when nodes $w_i$ and $w_j$ transmit simultaneously. Since $v_t$ has no other neighbors and the protocol is 1-shot, $v_t$ never gets the message.} 
   \label{collision}
   \end{center}
\end{figure}

Clearly, for every graph in $\mathcal{H}_{1}$, no node from $\overline{S}$ transmits at step $T+1$. This means that only nodes from $V(S)$ transmit at step $T+1$. Consequently, all the nodes in $\overline{S}_1$ may receive a message only from $v_S$ and thus the history (and action) of any node $w \in \overline{S}_1$ is the same for any graph in $\mathcal{H}_1^w$. Using the same argument as before, there exists at most one node $w_2$ such that $\pi(w_2,T+2,H_{T+1}(w_2)) = \verb|send|$. Then, we further refine the family of graphs to $\mathcal{H}_{2} = \mathcal{H}_{1} \setminus \mathcal{H}_1^{w_2}$ and set $\overline{S}_2 = \overline{S}_1 \setminus \{ w_2\}$ (or $\mathcal{H}_{2} = \mathcal{H}_{1}$, $\overline{S}_2 = \overline{S}_1$ if no node transmits at $T+2$). We may apply this argument repeatedly at least $n-i-1$ times, each time further refining the family of graphs. The final family $G_{i} = \mathcal{H}_{l}$ ($l \geq n-i-1$) consists of graphs where the only node of $\overline{S}_{l} = \{ w_{l+1} \}$ is connected to $v_S$; thus, every graph in $\mathcal{G}_{i}$ starts with the chain $S \circ w_{l+1}$. Moreover, $w_{l+1}$ does not transmit earlier than step $T+l+1$ in any graph of $\mathcal{G}_{i}$.
\end{proof}

\begin{thm}
For any 1-shot broadcasting protocol $\mathcal{P}$, there exists a graph $G_{\mathcal{P}}$ where $\mathcal{P}$ needs $\Omega(n^2)$ steps to complete broadcasting.  
\end{thm}

\begin{proof}
We will prove the lower bound using induction on the number of stages. Specifically, we will show family $\mathcal{G}_i$ includes only graphs starting with a chain of $i+1$ nodes and the last node of the chain does not transmit before step $1+ \sum_{j=1}^i (n-j)$. 

At the end of stage 0, the chain consists only of the source node, which transmits at step 1. Thus, the claim holds trivially. Using the induction hypothesis, we know that the family $\mathcal{G}_{i}$ includes only graphs which start with a chain of $i+1$ nodes and the last node transmits no earlier than step $1+ \sum_{j=1}^i (n-j)$. Applying lemma \ref{slow_n}, the construction refines the family $\mathcal{G}_{i}$ to $\mathcal{G}_{i+1}$, where every graph of $\mathcal{G}_{i+1}$ starts with a chain of $(i+1)+1$ nodes and the last node transmits no earlier than step $1+ \sum_{j=1}^i (n-j) + (n-i-1) =1+ \sum_{j=1}^{i+1} (n-j)$. 

After stage $n-3$, the family $\mathcal{G}_{n-3}$ includes only three graphs (for each configuration of the remaining two nodes). In one of the graphs (graph $G_{\mathcal{P}}$), the protocol needs one more step to complete broadcasting. Thus, in $G_{\mathcal{P}}$ the protocol completes broadcasting no earlier than step
$$ 1+\sum_{i=1}^{n-3} (n-i)+1 = \frac{n(n-1)}{2}-1$$   
\end{proof}

The proof of this theorem is constructive. Thus, for any 1-shot broadcasting protocol $\mathcal{P}$, we can actually construct a graph $G$ where $\mathcal{P}$ needs at least $\Omega(n^2)$ steps to complete broadcasting.

\end{document}

%% file: Settings.tex
\usepackage{ucs}
\usepackage[utf8x]{inputenc}
\usepackage{amsmath, amsthm,amssymb,amsfonts}    
\usepackage{graphicx}   
\usepackage{verbatim}   
\usepackage{color}      
\usepackage{subfigure}  
\usepackage{framed}

\usepackage[ruled, vlined]{algorithm2e}

\newtheorem{thm}{Theorem}[section]

\newtheorem{lem}[thm]{Lemma}

\newtheorem{definition}{Definition}
\newtheorem{theorem}{Theorem}

\usepackage[authoryear]{natbib}
\renewcommand{\cite}{\citep}
\newcommand{\calb}{\mathcal{B}}

\newcommand{\calg}{\mathcal{G}}

\newcounter{protocolcount}